\documentclass[]{gAPA2e}

\newcommand{\dsp}{\displaystyle}

\newcommand{\pl}{\partial}
\newcommand{\mc}{\mathcal}
\newcommand{\mf}{\mathbf}
\newcommand{\courier}{\fontfamily{pcr}\selectfont}
\newcommand{\D}{\ensuremath{\mathcal{D}}}
\newcommand{\Id}{\ensuremath{\mathcal{I}}}
\newcommand{\M}{\ensuremath{\mathcal{M}}}

\usepackage{xcolor}
\usepackage{verbatim}

\begin{document}
\doi{10.1080/0003681YYxxxxxxxx}
 \issn{1563-504X}
\issnp{0003-6811}
\jvol{00} \jnum{00} \jyear{2009} \jmonth{January}

\markboth{D. Poole and W. Hereman}{Applicable Analysis}

\articletype{}

\title{The homotopy operator method for symbolic integration by parts
and inversion of divergences with applications}

\author{Douglas Poole and Willy Hereman$^{\ast}$\thanks
{$^\ast$Corresponding author. Email: whereman@mines.edu}
       \\ \vspace{6pt}  
       \em{Department of Mathematical and Computer Sciences, 
           Colorado School of Mines, Golden, CO 80401-1887, USA}}

\maketitle

\begin{abstract}
Using standard calculus, explicit formulas for one-, two- and 
three-dimensional homotopy operators are presented.
A derivation of the one-dimensional homotopy operator is given.
A similar methodology can be used to derive the multi-dimensional versions.
The calculus-based formulas for the homotopy operators are easy to 
implement in computer algebra systems such as {\it Mathematica}, 
{\it Maple}, and {\rm REDUCE}.
Several examples illustrate the use, scope, and limitations of the homotopy 
operators.
The homotopy operator can be applied to the symbolic computation
of conservation laws of nonlinear partial differential equations (PDEs).
Conservation laws provide insight into the physical and mathematical properties
of the PDE.
For instance, the existence of infinitely many conservation laws establishes 
the complete integrability of a nonlinear PDE.
\begin{classcode}
Primary: 37J35, 37K40, 35Q51; 
Secondary: 68W30, 47J35, 70H06.
\end{classcode}
\end{abstract}
\vspace{-5mm}
\noindent
\section{Introduction}
In many applications one needs to integrate exact expressions involving 
several unspecified functions of multiple independent variables.
The ``homotopy operator" presented in this paper performs that task.
Indeed, we give a formula based on a homotopy that integrates one-dimensional
(1\,D) expressions, or inverts a total divergence in two dimensions (2\,D) 
or three-dimensions (3\,D).
The formulas are derived using a principle of the calculus of variations
and are written in the language of standard calculus.

The concept of homotopy, contributed to Poincar\'e, is certainly familiar to 
differential geometers and algebraic topologists \cite{Cohen73,Tu2008}.
Briefly, a homotopy is a continuous map, 
$T$: $X \times [0,1] \longrightarrow Y$ between two functions 
$\mf{u}$ and $\mf{u_0}$ (both are maps from a space $X$ to a space $Y),$ 
such that $T(\mf{x},0) = \mf{u_0}(\mf{x})$ and $T(\mf{x},1) = \mf{u}(\mf{x}).$
For example, 
$T(\mf{x},\lambda) = \lambda \mf{u}(\mf{x}) + (1-\lambda) \mf{u_0}(\mf{x})
= \mf{u_0}(\mf{x}) + \lambda ( \mf{u}(\mf{x}) - \mf{u_0}(\mf{x}) )$
is a homotopy.
If $\mf{u_0} = \mf{0}$ then $T(\mf{x},\lambda) = \lambda \mf{u}(\mf{x})$ 
expresses a scaling of $\mf{u}$ with a parameter $\lambda.$
The homotopy operator used in this paper is more elaborate for it integrates 
an expression with respect to an auxiliary variable $\lambda$ from $0$ to $1$.

Homotopy operators presumably first appeared in the inverse problem of 
the calculus of variations in works by Volterra \cite{Volterra13}.
They also appear in the proof of the converse of the Poincar\'{e} lemma, 
which states that closed differential $k$-forms are exact \cite{Flanders63},
but only on domains with certain topological assumptions such as a 
star-shaped domain \cite{Olver93}.
Once the domain is appropriately restricted, the proof of the converse of
Poincar\'{e}'s lemma requires the construction of a homotopy operator.
More generally, the construction of suitable homotopy operators is the key 
to exactness proofs for many complexes such as the de Rham complex and the 
bi-variational complex \cite{Olver93}.

One method to investigate the complete integrability 
\cite{AblowitzClarkson91} of a system of partial differential equations (PDEs)
is to determine whether the system has infinitely many conservation laws.
While developing a method for computing conservation laws of nonlinear 
PDEs in $(1+1)$ dimensions \cite{AdamsMS03,HeremanAll09}, the authors 
needed a technique to symbolically integrate expressions involving 
unspecified functions. 
{\it Mathematica}'s {\courier Integrate} function often failed to integrate 
such integrands, in particular, when transcendental functions were present. 
The computation of conservation laws of nonlinear PDEs in multiple space 
variables, requires a tool to invert total divergences 
\cite{HeremanAll05,HeremanIJQC06,Poole09}. 
The homotopy operator used in the proof of the exactness of the 
(bi-)variational complex \cite{Olver93} can do the required integrations.
The homotopy operator in \cite{Olver93} has since been translated into the
language of standard calculus and written into a more efficient algorithmic
form \cite{HeremanAll09,HeremanAll05,HeremanIJQC06,Poole09}.
Although the homotopy operator for one variable does not appear in work 
by Kruskal {\it et al}.\ \cite{KruskalAll70}, it takes only one extra step 
to derive it. 
Taking that step, it also became clear how to generalize the formula to cover 
multiple variables.

The calculus-based form given in \cite{HeremanAll09,HeremanAll05} makes 
the homotopy operator easy to apply by researchers not familiar with 
differential forms.
However, two issues must be addressed before the homotopy operator becomes 
a ready-to-use tool.
The first issue relates to the inversion of a total divergence,
$\mathrm{Div}^{-1}$, which does not have a unique answer.
In analogy with standard integration with $\D_x^{-1}$, where the result is 
up to an arbitrary constant, $\mathrm{Div}^{-1}$ is defined up to curl terms.
The homotopy operator will produce a particular choice for the curl terms
\cite{HeremanAll05,HeremanIJQC06}, often creating very large vectors.
Hence, we discuss an algorithm to remove the unwanted curl terms.
Secondly, the homotopy operator fails to work on expressions involving terms 
of degree zero.
A solution to this problem is also presented in this paper, thereby extending
the applicability of the homotopy operator.

With applications in mind far beyond the computation of conservation laws,
the homotopy operator has its own {\it Mathematica} code,
{\tt HomotopyIntegrator.m} \cite{PooleHereman09software1}.
To our knowledge, {\tt HomotopyIntegrator.m} is currently the only
full-fledged implementation of the homotopy operator in {\it Mathematica}.
In collaboration with Hereman \cite{HeremanAll05}, Deconinck and Nivala 
\cite{DeconinckNivala05,DeconinckNivala08} have recently developed similar 
code in {\it Maple}. 
Although applicable to exact as well as non-exact expressions, their code is 
restricted to the 1\,D case.
Starting from version 9 of {\it Maple}, the homotopy operator is now part of
the kernel of {\it Maple} to broaden the scope of its {\courier integrate}
function.
Anderson \cite{Anderson04} and Cheviakov \cite{Cheviakov07,Cheviakov09} offer 
implementations of the homotopy operator in {\it Maple}, not as stand-alone 
tools but as a component of broader software packages.

The multi-dimensional homotopy operator is an essential tool for the 
symbolic computation of fluxes of conservation laws of nonlinear PDEs 
involving multiple space variables. 
Most of such systems are not completely integrable for they only have a 
finite number of conservation laws \cite{Poole09}.
As an application, we compute a conservation law of the $(2+1)-$dimensional 
Zakharov-Kuznetsov equation \cite{ZakharovKuznetsov74}, 
which describes ion-acoustic solitons in magnetized plasmas.
Recently, we developed the {\it Mathematica} package 
{\tt ConservationLawsMD.m} \cite{PooleHereman09software2} that 
automates the computation of polynomial conservation laws of nonlinear 
polynomial PDEs involving multiple space variables and time.
The multi-dimensional homotopy operator is a key tool in that package.
\section{Preliminary Definitions}
Operations are carried out on differential functions,
$f(\mf{x},\mf{u}^{(M)}(\mf{x})),$ where $\mf{u}^{(M)}(\mf{x})$ 
denotes the dependent variable $\mf{u} = (u^1, \dots, u^j, \dots, u^N)$ 
and its partial derivatives (up to order $M)$ with respect to independent 
variable $\mf{x}.$
Although we allow variable coefficients, the differential functions should 
not contain terms that are functions of $\mf{x}$ only.
For simplicity of notation in the examples, we will denote 
$u^1, u^2, u^3,$ etc.\ by $u,v,w,$ etc.\ 

We consider only 1\,D, 2\,D, and 3\,D cases.
Therefore, the independent space variable, $\mf{x},$ will have at most three 
components.
Thus, $\mf{x}$ represents $x$, $(x,y),$ and $(x,y,z)$ in 1\,D, 2\,D, and 3\,D 
problems, respectively.
Partial derivatives are denoted by
$u_{k_1 x \, k_2 y \, k_3 z} = \frac{\pl^{k_1 + k_2 + k_3} u}
{\pl x^{k_1}\pl y^{k_2} \pl z^{k_3}}$.
For example, $\frac{\pl^5 u}{\pl x^3 \pl y^2}$ is written as $u_{3x\,2y}.$

For simplicity, we often abbreviate $f(\mf{x},\mf{u}^{(M)}(\mf{x}))$ by $f.$
The notation $f[\lambda \mf{u}]$ means that in $f$ one replaces 
$\mf{u}$ by $\lambda\,\mf{u}$, $\mf{u}_x$ by $\lambda\,\mf{u}_x$, 
and so on for all derivatives of $\mf{u},$ 
where $\lambda$ is an auxiliary parameter.
For example, if $f = x u^2 v_x + u_x v_{2x} \sin w$ then 
$f[\lambda \mf{u}] = 
\lambda^3 x u^2 v_x + \lambda^2 u_x v_{2x} \sin \lambda w.$

Two additional operators are needed before the homotopy operator can 
be defined.
The total derivative operator allows one to algorithmically compute
derivatives of differential functions based on the chain rule of 
differentiation.
\begin{definition}
With $\mf{x} = x$, the 1\,D total derivative operator $\D_x$ acting on 
$f = f(x,\mf{u}^{(M)}(x))$ is defined as 
%
\[
\D_x f = \frac{\pl f}{\pl x} + \sum_{j=1}^N \sum_{k=0}^{M_1^j}
u_{(k+1) x}^j \frac{\pl f}{\pl u_{k x}^j},
\]
where $M_1^j$ is the order of $f$ in component $u^j$ and 
$M = {\rm max}_{j = 1, \dots, N} M_1^j.$
The partial derivative, $\frac{\pl }{\pl x},$ acts on any $x$ that appears
explicitly in $f$, but not on $u^j$ or any partial derivatives of $u^j.$
The total derivative operators in 2\,D and 3\,D are defined analogously.
For example, the 3\,D total derivative operator with respect to $x$ is 
\[
\D_x f = \frac{\pl f}{\pl x}
+ \sum_{j=1}^N \sum_{k_1=0}^{M_1^j} \sum_{k_2=0}^{M_2^j}
\sum_{k_3=0}^{M_3^j} u_{(k_1+1) x \, k_2 y \, k_3 z}^j
\frac{\pl f}{\pl u_{k_1 x \, k_2 y \, k_3 z}^j},
\]
where $M_1^j$, $M_2^j$ and $M_3^j$ are the orders of $f$ for component $u^j$
with respect to $x$, $y$, and $z$, respectively 
\end{definition}
Obviously, $\D_y$ and $\D_z$ can be defined analogously.
Note that for the 3\,D case, 
$M$ is the maximum order for derivatives on $u^j,\, j = 1, \ldots, N.$

The Euler operator, also known as the variational derivative, is one of 
the most important tools in the calculus of variations.
\begin{definition}
The Euler operator for the 1\,D case where $f = f(x, \mf{u}^{(M)}(x))$ is 
\begin{eqnarray}
\label{zeroeulerscalarux1D}
\mc{L}_{u^j (x)} f
&=& \sum_{k=0}^{M_1^j} (-\D_x)^k \frac{\pl f}{\pl u_{kx}^j}
\nonumber \\ 
&=& \frac{\pl f}{\pl u^j} - \D_x \frac{\pl f}{\pl u_x^j}
+ \D_x^2 \frac{\pl f}{\pl u_{2x}^j}
- \D_x^3 \frac{\pl f}{\pl u_{3x}^j} + \cdots
+ (-\D_x)^{M_1^j} \frac{\pl f}{\pl u_{M_1^j x}^j},
\end{eqnarray}
$j = 1, \dots, N.$
The 2\,D and 3\,D Euler operators (variational derivatives) are defined 
analogously.
For example, in the 2\,D case where $\mf{x} = (x,y)$ and 
$f = f(\mf{x},\mf{u}^{(M)}(\mf{x}))$, the Euler operator is defined as
\begin{equation}
\label{zeroeulerscalarux2D}
\mc{L}_{u^j(x,y)} f
= \sum_{k_1=0}^{M_1^j} \sum_{k_2=0}^{M_2^j} (-\D_x)^{k_1} (-\D_y)^{k_2}
\frac{\pl f}{\pl u_{k_{1}x\,k_{2}y}}, \quad j = 1, \dots, N.
\end{equation}
\end{definition}
In this paper we use the Euler operators to test if differential functions 
are exact.
\begin{definition}
Let $f = f(\mf{x}, \mf{u}^{(M)}(\mf{x}))$ be a differential function
of order $M.$ 
When $\mf{x} = x$, $f$ is called {\it exact} if there exists a differential 
function $F(x, \mf{u}^{(M-1)}(x))$ such that $f = D_x F$.  
When $\mf{x} = (x, y)$ or $\mf{x} = (x,y,z)$,
$f$ is {\it exact} if there exists a differential vector function
$\mf{F}(\mathbf{x},\mf{u}^{(M-1)}(\mf{x}))$ such that 
$f = \mathrm{Div}\,\mf{F}.$
In the 2\,D case, $\mf{F} = (F^1,F^2)$ and 
$\mathrm{Div}\,\mf{F} = \D_x F^1 + \D_y F^2.$
The definition of $\mathrm{Div}$ in 3\,D is analogous.
\end{definition}
\begin{theorem}
\label{zeroeulerexact}
A differential function $f = f(\mf{x},\mf{u}^{(M)}(\mf{x}))$ is exact 
if and only if 
$\mc{L}_{\mf{u}(\mf{x})} f \equiv \mf{0}.$
Here, $\mf{0}$ is the vector $(0,0,\cdots,0)$ which has $N$ components
matching the number of components of $\mf{u}.$
\end{theorem}
\noindent 
A proof for Theorem \ref{zeroeulerexact} is given in \cite{Poole09}.
\vspace{-3mm}
\noindent
\section{Integration by Parts Using the One-Dimensional Homotopy Operator}
In this section we show how to compute $F = \D_x^{-1} f = \int f \, dx$ 
when $f$ is exact.
Direct integration by parts is cumbersome and prone to errors if $f$ contains 
high-order derivatives of several dependent variables.
The 1\,D homotopy operator provides an alternate method for integrating exact 
differential functions (with one independent variable) and multiple dependent 
variables, often eliminating the need for integration by parts.
\begin{definition}
\label{oneDhomotopyoperator}
Let $\mf{x} \!=\! x$ be the independent variable and 
$f = f(x, \mf{u}^{(M)}(x))$ be an exact differential function, 
i.e.\ there exist a function $F$ such that $F = \D_x^{-1} f.$ 
Thus, $F$ is the integral of $f.$
The 1\,D homotopy operator is defined as
\begin{equation}
\label{1Dhomotopy}
\mc{H}_{\mf{u}(x)} f =  
\int_{\lambda_0}^1 
\left( 
\sum_{j=1}^{N} 
\mc{I}_{u^j(x)} f \right)[\lambda {\mf{u}}]\,\frac{d \lambda}{\lambda},
\end{equation}
where $\mf{u} = (u^1, \dots, u^j, \dots, u^N).$
The integrand, $\mc{I}_{u^j(x)} f,$ is defined as
\begin{equation}
\label{1Dintegrand}
\mc{I}_{u^j(x)} f 
= \sum_{k=1}^{M_1^j} \left( \sum_{i=0}^{k-1}{u_{ix}^j}
\left( -\D_x \right)^{k-(i+1)} \right) \frac{\pl f}{\pl {u_{kx}^j}},
\end{equation}
where $M_1^j$ is the order of $f$ in dependent variable $u^j$ with respect 
to $x$.
\end{definition}
The usual homotopy operator with $\lambda_0 = 0$ applies when 
(\ref{1Dhomotopy}) converges.
However, due to a possible singularity at $\lambda = 0,$ 
(\ref{1Dhomotopy}) might diverge for $\lambda_0 \to 0.$ 
This can occur with rational as well as irrational integrands.
In such cases, one can take $\lambda_0 \to \infty$ or, alternatively, 
evaluate the indefinite integral and let $\lambda \to 1.$
As shown in \cite{DeconinckNivala08}, the homotopy operator with appropriately
selected limits remains a valid tool for integrating non-polynomial functions.

In the next section we prove that $F = \mc{H}_{\mf{u}(x)} f.$
For now, we illustrate the application of the homotopy operator formulas 
(\ref{1Dhomotopy}) and (\ref{1Dintegrand}) with an example. 
\begin{example}
Let $(u^1,u^2) = (u,v)$ and take $f(x,\mf{u}^{(3)}(x)) 
= u^2 + 2 x u u_x + u_x v_{3x} + u_{2x} v_{2x} -3 v_x^2 v_{2x}.$
Applying the Euler operator (\ref{zeroeulerscalarux1D}) separately for 
$u$ and $v,$ with $M_1^1 = 2$ and $M_1^2 = 3,$ one readily verifies that
$\mc{L}_{u(x)} f \equiv 0$ and ${L}_{v(x)} f \equiv 0.$
Thus, $f$ is exact.
To find $F = \int f\,dx$, using (\ref{1Dintegrand}), first compute the 
integrands, 
\begin{eqnarray*}
\mc{I}_{u(x)} f  
&=& ( u \Id )(\frac{\pl f}{\pl u_x}) 
+ (u_x \Id - u \D_x) (\frac{\pl f}{\pl u_{2x}}) 
= u (2 x u + v_{3x}) + (u_x \Id - u \D_x) (v_{2x}) 
\nonumber \\  
&=& 2 x u^2 + u_x v_{2x}, 
\nonumber \\ 
\mc{I}_{v(x)} f  
&=& ( v \Id )(\frac{\pl f}{\pl v_x}) 
+ (v_x \Id - v \D_x) (\frac{\pl f}{\pl v_{2x}}) 
+ (v_{2x} \Id - v_x \D_x + v \D_x^2) (\frac{\pl f}{\pl v_{3x}}) 
\nonumber \\ 
&=& -v (6 v_x v_{2x}) + (v_x \Id - v \D_x) (u_{2x} - 3 v_x^2)
+ (v_{2x} \Id - v_x \D_x + v \D_x^2) (u_x) 
\nonumber \\ 
&=& u_x v_{2x} - 3 v_x^3,
\end{eqnarray*}
where $\Id$ is the identity operator.
Next, sum the integrands, replace $u, u_x, v_x$ and $v_{2x}$ with 
$\lambda u, \lambda u_x, \lambda v_x$, and $\lambda v_{2x},$ respectively, 
divide by $\lambda,$ and integrate with respect to $\lambda.$ 
In detail, 
\begin{eqnarray*}
F = \mc{H}_{\mf{u}(x)} f &=& \int_0^1 
\left( \mc{I}_{u(x)} f + \mc{I}_{v(x)} f \right)
[\lambda \mf{u}]\,\frac{d \lambda}{\lambda} 
\nonumber \\ 
&=& \int_0^1 ( 2 \lambda x u^2 + 2 \lambda u_x v_{2x} 
    - 3 \lambda^2 v_x^3 )\,d \lambda 
\nonumber \\ 
&=& \left. ( \lambda^2 x u^2 + \lambda^2 u_x v_{2x} 
    - \lambda^3 v_x^3 )\right|_{0}^{1} 
\nonumber \\ 
&=& x u^2 + u_x v_{2x} - v_x^3.
\end{eqnarray*}
Clearly, $\D_x F = f.$
Here $\lambda_0 = 0$ was used since (\ref{1Dhomotopy}) converged for 
$\lambda_0 \to 0.$
\end{example}
For polynomial differential expressions, (\ref{1Dhomotopy}) replaces 
integration by parts (in $x)$ with a few differentiations following by a 
standard integration (of a polynomial) with respect to an auxiliary parameter
$\lambda.$
The second example illustrates the integration of a rational differential 
function where (\ref{1Dhomotopy}) diverges for $\lambda_0 \to 0.$
\begin{example}
Let $(u^1,u^2) = (u,v)$ and take
$f(x,\mf{u}^{(1)}(x)) = (v u_x + u v_x)/(uv)^2.$
Using (\ref{1Dintegrand}), 
$\mc{I}_{u(x)} f = \frac{1}{u v}$ and $\mc{I}_{v(x)} f = \frac{1}{u v}.$
Evaluation of (\ref{1Dhomotopy}) gives
\begin{eqnarray*}
F = \mc{H}_{\mf{u}(x)} f 
&=& \int_{\lambda_0}^1 \frac{2}{\lambda^3 u v}\,d \lambda 
= - \left. \frac{1}{\lambda^2 u v}\right|_{\lambda_0}^{1}.
\end{eqnarray*}
Obviously, $\lambda_0 \to 0$ would cause the integral to diverge.
Instead, take $\lambda_0 \to \infty$ so that 
$F = \mc{H}_{\mf{u}(x)} f = - \frac{1}{u v}.$
Alternatively, when a singularity occurs at $\lambda = 0$, one 
could compute the indefinite integral and let $\lambda \to 1.$
In either case, $\D_x F = f.$
\end{example}
\vspace{-3mm}
\noindent
\section{The One-Dimensional Homotopy Operator}
\label{sec1Dhomotopy}
The 1\,D homotopy operator in Definition \ref{oneDhomotopyoperator} is an 
expanded version of the one (in terms of higher-Euler operators) given in 
\cite{HeremanAll05,DeconinckNivala05,DeconinckNivala08,Cheviakov09}. 
Although they give the same result \cite{Poole09,HeremanDeconinckPoole06}, 
the expanded homotopy operator formula is computationally more efficient 
for two reasons: the number of times that the total derivative is applied 
is significantly smaller and the combinatorial coefficients have been 
eliminated.

To prove that the homotopy operator (\ref{1Dhomotopy}) does indeed integrate 
an exact differential expression, some extra definitions and theorems 
are needed.
We first define a {\it degree operator}, $\M.$
Application of $\M$ to a differential function yields that function multiplied 
by its degree; hence its name.
\begin{definition}
\label{moperatordefn}
Let $\mf{x} = x$ be the independent variable.
The degree operator $\M$ acting on a differential function
$f = f(x, \mf{u}^{(M)}(x))$ is defined \cite{KruskalAll70} as
\begin{equation}
\label{moperator}
\M f = \sum_{j=1}^N \sum_{i=0}^{M_1^j} u_{ix}^j \frac{\pl f}{\pl u_{ix}^j},
\end{equation}
where $f$ has order $M_1^j$ in $u^j$ with respect to $x$.
\end{definition}
The next example shows how the degree operator works. 
\begin{example}
\label{moperatorexpl}
Let $\mf{u} = (u^1, u^2) = (u,v)$ and 
$f(x, \mf{u}^{(5)}(x)) 
= \left( u \right)^p \left( v_{2x} \right)^q \left( u_{5x} \right)^r$,
where $p$, $q$, and $r$ are nonzero rational numbers.
Using (\ref{moperator}), 
\begin{eqnarray*}
\M f &\!=\!& u \frac{\pl}{\pl u} \left( u \right)^p
\left( v_{2x} \right)^q \left( u_{5x} \right)^r
\!+\! u_{5x} \frac{\pl}{\pl u_{5x}} \left( u \right)^p
\left( v_{2x} \right)^q \left( u_{5x} \right)^r
\!+\! v_{2x} \frac{\pl}{\pl v_{2x}}  \left( u \right)^p
\left( v_{2x} \right)^q \left( u_{5x} \right)^r
\nonumber \\ 
&=& p u \left( u \right)^{p-1}\left( v_{2x} \right)^q \left( u_{5x} \right)^r
+r u_{5x} \left( u \right)^p\left( v_{2x} \right)^q\left( u_{5x} \right)^{r-1}
+q v_{2x} \left( u \right)^p\left( v_{2x} \right)^{q-1}\left( u_{5x} \right)^r
\nonumber \\ 
&=& (p+q+r)\left( u \right)^p \left( v_{2x} \right)^q \left( u_{5x} \right)^r.
\end{eqnarray*}
Note that the total degree of $f$ has become a factor.
Compare this with $x(x^n)^{\prime} = x (n x^{n-1}) = n x^n$ 
$(n$ is a nonzero rational) 
in 1\,D calculus. 
\end{example}
It is trivial to prove that $\M$ is a linear operator. 
Less straightforward is finding the kernel $(\mathrm{Ker})$ of $\M.$
For what follows, we need the notion of homogeneity.
\begin{definition}
A differential function $f = f( \mf{x}, \mf{u}^{(M)}(\mf{x}))$ is called 
{\it homogeneous} of degree $p$ in $\mf{u}$ (and its derivatives) 
if $f[\lambda \mf{u}] = \lambda^p f.$
\end{definition}
Note the analogy with the definition of a homogeneous function of several 
variables. 
For example for the two-variables case, $f(x,y)$ is called homogeneous of 
degree $p$ if $f(\lambda x, \lambda y) = \lambda^p f(x,y).$
Examples are $f(x,y) = a x^2 + b x y + c y^2$ of degree 2 and 
$f(x,y) = 1/\sqrt{ a x^2 + b x y + c y^2}$ of degree $-1$ 
$(a,b,$ and $c$ are constants).
\begin{theorem}
\label{TheoremkernelM}
If $f = \frac{k}{\ell}$ where $k = k(x,\mf{u}^{(M)}(x))$ and 
$\ell = \ell(x,\mf{u}^{(M)}(x))$ are homogeneous differential functions of 
the same degree (so, $f$ is of degree $0)$ then $\M f = 0.$
\end{theorem}
\begin{proof}
Suppose that $k$ and $\ell$ are homogeneous differential functions of 
degree $p.$ 
Obviously, $\M k = p k$ and $\M \ell = p \ell.$
Hence, $\M (\frac{k}{\ell}) 
= \frac{\ell (\M k) - k (\M \ell)}{\ell^2} 
= \frac{ p \ell k - p k \ell}{\ell^2} = 0.$
\end{proof}
Conversely, one could ask
``Which differential functions are in the kernel of the degree operator?" 
For simplicity, consider the case where $f(u, v, u_x, v_x).$ 
Any $f \in {\rm Ker}\,\M$ must satisfy $\M f = 0,$ or explicitly, 
\begin{equation}
\label{Monf}
u \frac{\partial f}{\partial u} 
+ u_x \frac{\partial f}{\partial u_x} 
+ v \frac{\partial f}{\partial v}
+ v_x \frac{\partial f}{\partial v_x} = 0.
\end{equation}
The jet variables $u, v, u_x, $ and $v_x$ are independent.
Hence, (\ref{Monf}) is a linear first-order PDE which can be solved 
with the method of the characteristics\footnote{We are indebted to 
Mark Hickman for this argument and subsequent derivation.}.
Solving the first-order system of ODEs, 
%
\[
\frac{d u}{d \tau}   = u,   \qquad 
\frac{d u_x}{d \tau} = u_x, \qquad 
\frac{d v}{d \tau}   = v,   \qquad 
\frac{d v_x}{d \tau} = v_x, 
\]
%
where $\tau$ parameterizes the characteristic curves, yields
\begin{equation}
\label{solchareqs}
u   = c_1 {\rm e}^{\tau}, \qquad
u_x = c_2 {\rm e}^{\tau}, \qquad
v   = c_3 {\rm e}^{\tau}, \qquad
v_x = c_4 {\rm e}^{\tau}, 
\end{equation}
where the $c_i$ are arbitrary constants.
The first equation in (\ref{solchareqs}) determines 
${\rm e}^{\tau} = \frac{u}{c_1}.$
Using it to eliminate $\tau$ from the remaining equations of 
(\ref{solchareqs}), yields
\[
\label{char}
u_x = \frac{c_2}{c_1} u = C_1\,u, \qquad
v   = \frac{c_3}{c_1} u = C_2\,u, \qquad
v_x = \frac{c_4}{c_1} u = C_3\,u,
\]
where $C_1, C_2,$ and $C_3$ are arbitrary constants.
Thus, the general solution of (\ref{Monf}) is 
$f(C_1, C_2, C_3) = f(\frac{u_x}{u}, \frac{v}{u}, \frac{v_x}{u} ),$ 
where the functional form of $f$ is arbitrary.
\begin{example}
By construction, 
\begin{equation}
\label{examplekernelMtest}
f = \mathrm{D}_x \left( \frac{u+v}{u-v} \right)
  = \frac{2 (u v_x - u_x v)}{(u-v)^2}
\end{equation}
is exact. 
Furthermore, $\M f = 0$ by Theorem \ref{TheoremkernelM}.
Given $f(u,u_x,v,v_x),$ not necessarily rational, one can {\it a priori} 
test whether or not $f \in {\rm Ker}\,\M.$ 
Indeed, applying the replacement rules, 
%
$ u   \rightarrow u, \;
u_x \rightarrow \mu u, \;
v   \rightarrow \mu u, \;
v_x \rightarrow \mu u, 
$
to $f$ should yield an expression which is independent of $u, u_x, v,$
and $v_x.$
Applying the replacement rules to $f$ in (\ref{examplekernelMtest}) gives 
$\frac{2 \mu}{1 - \mu},$ which only depends on $\mu.$
\end{example}
The argument above can be extended to any differential function $f$ 
no matter which jet variables are present.
The above ``kernel test" can be performed on exact as well as non-exact 
differential functions. 

Next, the inverse of the degree operator will be derived, which requires 
the use of a homotopy \cite{HeremanDeconinckPoole06}.
To avoid a potentially diverging integral at $0$, the homotopy is defined 
on $[\lambda_0, 1]$. 
To begin with, consider only differential functions with one term, 
e.g., $u u_x$, $x u_x$, $u_x/u, u_x^2/(\sqrt{u^4 +v_x^4}),$ or take $f$ in 
Example \ref{moperatorexpl}.
\begin{theorem}
\label{inverseMtheorem}
Let $f(x, \mf{u}^{(M)}(x))$ be any single term differential function 
in 1\,D and let $g(x, \mf{u}^{(M)}(x)) = \M f(x, \mf{u}^{(M)}(x)),$
where $\M f(x, \mf{u}^{(M)}(x)) \ne 0.$
Then,
\begin{equation}
\label{inverseMthm}
\M^{-1} g(x, \mf{u}^{(M)}(x))
= \int_{\lambda_0}^1 g[\lambda \mf{u}]\,\frac{d \lambda}{\lambda}.
\end{equation}
\end{theorem}
\begin{proof}
If $g(x, \mf{u}^{(M)}(x))$ has order $M_1^j$ in $u^j$ with respect to $x$, 
then $g[\lambda \mf{u}]$ also has order $M_1^j$ in $u^j$.
Furthermore, using the chain rule, 
\begin{equation}
\label{dlambdaofg}
\frac{d}{d \lambda} g[\lambda \mf{u}] 
= \sum_{j=1}^N \sum_{i=0}^{M_1^j}
\frac{\pl g[\lambda \mf{u}]}{\pl \lambda u^j_{ix}}
\frac{d \lambda u^j_{ix}}{d \lambda} 
= \frac{1}{\lambda} \sum_{j=1}^N
\sum_{i=0}^{M_1^j} u^j_{ix} \frac{\pl g[\lambda \mf{u}]}{\pl u^j_{ix}}
= \frac{1}{\lambda} {\M} g[\lambda \mf{u}],
\end{equation}
by the definition of $\M.$  
Integrate both sides of (\ref{dlambdaofg}) with respect to $\lambda$ yields
\begin{eqnarray}
\label{laststepMinverse}
\int_{\lambda_0}^1 \frac{d}{d \lambda} g[\lambda \mf{u}]\,d \lambda
&=& \int_{\lambda_0}^1 {\M} \frac{g[\lambda \mf{u}]}{\lambda}\,d \lambda, 
\nonumber \\ 
\left. {g \left[ \lambda \mf{u} \right] }\right|_{\lambda_0}^{1}
&=& {\M} \int_{\lambda_0}^1 g[\lambda \mf{u}]\,\frac{d \lambda}{\lambda}, 
\nonumber \\ 
g(x, \mf{u}^{(M)}(x)) - g[\lambda_0 \mf{u}] &=& {\M}
\int_{\lambda_0}^1 g[\lambda \mf{u}]\,\frac{d \lambda}{\lambda}.
\end{eqnarray}
To get (\ref{inverseMthm}), $g[\lambda_0 \mf{u}]$ must be $0.$
This will affect the choice of $\lambda_0.$
Indeed, using (\ref{moperator}),
%
\begin{equation}
\label{evaluategoflambdau} 
g[\lambda_0 \mf{u}] 
= \sum_{j=1}^N \sum_{i=0}^{M_1^j} \lambda_0 u_{ix}^j
\frac{\pl f[\lambda_0 \mf{u}]}{\pl (\lambda_0 u_{ix}^j)} 
= \lambda_0 \sum_{j=1}^N \sum_{i=0}^{M_1^j} u_{ix}^j
\frac{\pl f[\lambda_0 \mf{u}]}{\pl (\lambda_0 u_{ix}^j)}.
\end{equation}
Depending on the form of $f,$ there are two choices for $\lambda_0$ that 
make $g[\lambda_0 \mf{u}] = 0.$
\renewcommand{\labelenumi}{\roman{enumi}.}
\begin{enumerate}
\item 
If $g[\lambda_0 \mf{u}]$ in (\ref{evaluategoflambdau}) is an expression in 
fractional form where $\lambda_0$ is a factor in the {\it denominator}, 
then let $\lambda_0 \to \infty$ to get $g[\lambda_0 \mf{u}] = 0.$
The integral on the right-hand side of (\ref{laststepMinverse})
will then converge for $\lambda_0 \to \infty$.

\item For all other cases, provided $\lambda_0$ does not drop out of 
(\ref{evaluategoflambdau}) after simplification, set $\lambda_0 = 0$ 
to get $g[\lambda_0 \mf{u}] = 0.$
In the simplest case, when $f$ (and therefore also $g)$ is a homogeneous 
{\it monomial} of degree $p$, then $g[\lambda \mf{u}]$ has as factor 
$\lambda^p$ and, consequently, $g[\lambda_0 \mf{u}] = 0$ for $\lambda_0 = 0.$
\end{enumerate}
With $g[\lambda_0 \mf{u}] = 0$, apply $\M^{-1}$ to both sides of
(\ref{laststepMinverse}) to get (\ref{inverseMthm}).
\end{proof}
\renewcommand{\labelenumi}{\arabic{enumi}.}
Note that Theorem \ref{inverseMtheorem} excludes the case where 
$\M\,f(x, \mf{u}^{(M)}(x)) = 0.$
If $f \in \mathrm{Ker}\,\M$, then $g(x, \mf{u}^{(M)}(x)) = 0$ and 
Theorem \ref{inverseMtheorem} becomes trivial.
In view of Theorem \ref{TheoremkernelM}, the homotopy operator cannot 
be applied to expressions of degree $0.$ 
In particular, the homotopy operator would incorrectly handle integrands 
involving ratios of polynomial or irrational differential functions of 
like degree.
In Section \ref{seckerM}, we provide a method to overcome this 
shortcoming of the homotopy operator.

To show that (\ref{inverseMthm}) does indeed invert the $\M$ operator, 
in the next example we apply $\M^{-1}$ to the result of 
Example \ref{moperatorexpl}.
\begin{example}
Let $\mf{u} = (u^1,u^2) = (u,v)$ and 
$g(x, \mf{u}^{(M)}(x)) = (p+q+r)$
$\left( u \right)^p \left( v_{2x} \right)^q \left( u_{5x} \right)^r$.
Then, 
\begin{eqnarray}
\label{minverseresult}
\M^{-1} g(x, \mf{u}^{(M)}(x)) 
&=& \int_0^1 (p+q+r)\left( \lambda u \right)^p \left( \lambda v_{2x} \right)^q
\left( \lambda u_{5x} \right)^r \frac{d\,\lambda}{\lambda} 
\nonumber \\ 
&=& (p+q+r) \left(  u \right)^p \left(  v_{2x} \right)^q
\left( u_{5x} \right)^r \int_0^1 \lambda^{p+q+r-1} d\,\lambda
\nonumber \\ 
&=& (p+q+r) \left(  u \right)^p \left(  v_{2x} \right)^q
\left( u_{5x} \right)^r
\left. \Big( \frac{\lambda^{p+q+r}}{p+q+r} \Big)\right|^{1}_{0}
\nonumber \\ 
&=& \left(  u \right)^p \left(  v_{2x} \right)^q \left(  u_{5x} \right)^r.
\end{eqnarray}
As with all polynomial integrands, we took $\lambda_0 = 0$ as the lower 
limit of the homotopy integral.
Note that (\ref{minverseresult}) is identical to $f$ given in 
Example \ref{moperatorexpl}.
\end{example}
\begin{remark}
Let $f(x, \mf{u}^{(M)}(x)) = f_1 + f_2 + \cdots + f_i + \cdots + f_q$,
where each $f_i = f_i(x, \mf{u}^{(M)}(x))$ is a single-term differential 
function (not necessarily a monomial).  
Theorem \ref{inverseMtheorem} holds for $f$ provided that $\M f_i \ne 0$ 
for $i = 1, \dots, q.$
\end{remark}
%
%
Next, we establish the commutation of the degree operator (and its inverse) 
with the total derivative operator.
%
\begin{lemma}
\label{MDcommute1D}
$\M \D_x = \D_x \M$ and $\M^{-1} \D_x = \D_x \M^{-1}$.
\end{lemma}
\begin{proof}
A proof is given in \cite{Poole09}.
\end{proof}
Now follows the key theorem for this Section which states that the 1\,D 
homotopy operator does indeed integrate an exact differential function.
%
\begin{theorem}
\label{homotopy1Dtheorem}
Let $f = f(x, \mf{u}^{(M)}(x))$ be exact, i.e.\ $\D_x F = f$ for some 
differential function $F(x, \mf{u}^{(M-1)}(x)).$
Then $F = \D_x^{-1} f = \mc{H}_{\mf{u}(x)} f.$
\end{theorem}
%
\begin{proof}
The proof for the scalar case $(\mf{u} = u)$ in \cite{HeremanDeconinckPoole06} 
can be generalized to the vector case.
Indeed, first consider a fixed component $u^j$ of $\mf{u}$ and multiply
${\mc{L}}_{u^j(x)} f$ by $u^j$ to restore the degree. 
Subsequently, split off $\dsp u^j \frac{\pl f}{\pl u^j}.$ 
This yields
%
\begin{equation}
\label{proofhomotopyuxp1}
u^j {\mc{L}}_{u^j(x)} f
= u^j \sum_{k=0}^{M_1^j} (-\D_x)^k \frac{\pl f}{\pl u^j_{kx}} 
= u^j \frac{\pl f}{\pl u^j} + u^j \sum_{k=1}^{M_1^j} 
(-\D_x)^k \frac{\pl f}{\pl u^j_{kx}}. 
\end{equation}
Next, integrate the last term by parts and split off 
$\dsp u_x^j \frac{\pl f}{\pl u^j_{x}}.$
Continue this process until
$\dsp u_{M_1^j x}^j \frac{\pl f}{\pl u_{M_1^j x}^j}$ is split off.
Continuing with (\ref{proofhomotopyuxp1}), the computations proceed as 
follows:
\vspace*{-0.3cm}
\begin{eqnarray*}
\label{proofhomotopyuxp1bis}
u^j {\mc{L}}_{u^j(x)} f
&=& u^j \frac{\pl f}{\pl u^j} - \D_x \left(u^j \sum_{k=1}^{M_1^j} 
(-\D_x)^{k-1} \frac{\pl f} {\pl u^j_{kx}} \right) + u^j_x
\sum_{k=1}^{M_1^j}(-\D_x)^{k-1} \frac{\pl f}{\pl u^j_{kx}} 
\nonumber \\ 
&=& u^j \frac{\pl f}{\pl u^j} + u^j_x
\frac{\pl f}{\pl u^j_{x}} - \D_x \left( u^j
\sum_{k=1}^{M_1^j} (-\D_x)^{k-1} \frac{\pl f}{\pl u^j_{kx}} \right)
+ u^j_x \sum_{k=2}^{M_1^j}(-\D_x)^{k-1} \frac{\pl f}{\pl u^j_{kx}} 
\nonumber 
\end{eqnarray*}
\vspace{-3mm}
\noindent
\begin{eqnarray*}
\label{proofhomotopyuxp3}
\phantom{u^j \mc{L}_{u^j(x)} f}
&=& u^j \frac{\pl f}{\pl u^j} + u_x^j \frac{\pl f}{\pl u_{x}^j}
- \D_x \left( u^j \sum_{k=1}^{M_1^j} (-\D_x)^{k-1}
\frac{\pl f}{\pl u_{kx}^j} + u^j_x \sum_{k=2}^{M_1^j}(-{\D}_x)^{k-2}
\frac{\pl f}{\pl u_{kx}^j} \right) 
\nonumber \\ 
&& {} + u_{2x}^j \sum_{k=2}^{M_1^j}(-\D_x)^{k-2} \frac{\pl f}{\pl u_{kx}^j} 
\end{eqnarray*}
\vspace{-2mm}
\noindent
\begin{eqnarray*}
\label{proofhomotopyuxp4}
\phantom{u^j \mc{L}_{u^j(x)} f}
&=& \ldots 
\nonumber \\
&=& u^j \frac{\pl f}{\pl u^j} + u_x^j \frac{\pl f}{\pl u_{x}^j} 
+ \ldots + u_{M_1^jx}^j \frac{\pl f}{\pl u_{M_1^jx}^j} 
- \D_x \left( u^j \sum_{k=1}^{M_1^j} (-\D_x)^{k-1}
\frac{\pl f}{\pl u_{kx}^j} \right. 
\nonumber \\ 
&& \left. {} + u_x^j \sum_{k=2}^{M_1^j} (-\D_x)^{k-2}
\frac{\pl f}{\pl u_{kx}^j} + \ldots + u_{(M_1^j-1)x}^j
\sum_{k=M_1^j}^{M_1^j} (-{\D}_x)^{k-M_1^j} \frac{\pl f}{\pl u_{kx}^j } \right)
\end{eqnarray*}
\vspace{-2mm}
\noindent
\begin{eqnarray}
\label{proofhomotopyuxp5}
\phantom{u^j \mc{L}_{u^j(x)} f}
&=& \sum_{i=0}^{M_1^j} u_{ix}^j \frac{\pl f}{\pl u_{ix}^j} 
- \D_x \left( \sum_{i=0}^{M_1^j-1} u_{ix}^j \sum_{k=i+1}^{M_1^j}
(-\D_x)^{k-(i+1)} \frac{\pl f}{\pl u_{kx}^j } \right) 
\nonumber \\ 
&=& \sum_{i=0}^{M_1^j} u_{ix}^j \frac{\pl f}{\pl u_{ix}^j} 
- \D_x \left( \sum_{k=1}^{M_1^j} \left( \sum_{i=0}^{k-1}
u_{ix}^j (-\D_x)^{k-(i+1)} \right) \frac{\pl f}{\pl u_{kx}^j} \right).
\end{eqnarray}
Now, sum (\ref{proofhomotopyuxp5}) over all components $u^j$ to get
\vspace{-2mm}
\noindent
\begin{eqnarray*}
\sum_{j=1}^N u^j 
{\mc{L}}_{u^j(x)} f
&=& \sum_{j=1}^N \sum_{i=0}^{M_1^j} u_{ix}^j
\frac{\pl f}{\pl u_{ix}^j} - \sum_{j=1}^N
\D_x \left( \sum_{k=1}^{M_1^j} \left( \sum_{i=0}^{k-1} u_{ix}^j
(-\D_x)^{k-(i+1)} \right) \frac{\pl f}{\pl u_{kx}^j} \right) 
\nonumber \\ 
&=& \M f - \D_x \left( \sum_{j=1}^N \sum_{k=1}^{M_1^j} 
\left( \sum_{i=0}^{k-1} u_{ix}^j (-\D_x)^{k-(i+1)} \right) 
\frac{\pl f}{\pl u_{kx}^j} \right). 
\nonumber
\end{eqnarray*}
Since $f$ is exact, by Theorem \ref{zeroeulerexact}
$\mc{L}_{u^j(x)} f \equiv 0$ for $j = 1, \ldots, N,$
which implies that $\sum_{j=1}^N \mc{L}_{u^j(x)} f \equiv 0.$
Hence,
\begin{equation}
\label{mf}
{\M} f = \D_x \left( \sum_{j=1}^N
\sum_{k=1}^{M_1^j} \left( \sum_{i=0}^{k-1} u_{ix}^j
(-\D_x)^{k-(i+1)} \right) \frac{\pl f}{\pl u_{kx}^j} \right).
\end{equation}
Apply $\M^{-1}$ to both sides and replace $\M^{-1} \D_x$ by $\D_x \M^{-1}$
using Lemma \ref{MDcommute1D}.
Thus,
\begin{equation}
\label{expressionf}
f = \D_x \left( \M^{-1} \sum_{j=1}^N
\sum_{k=1}^{M_1^j} \left( \sum_{i=0}^{k-1} u_{ix}^j
(-\D_x)^{k-(i+1)} \right) \frac{\pl f}{\pl u_{kx}^j} \right).
\end{equation}
Apply $\D_x^{-1}$ to both sides of (\ref{expressionf}) and use 
(\ref{inverseMthm}) to obtain
\begin{equation}
\label{expressionF}
\D_x^{-1} f = \int_{\lambda_0}^{1} \left( \sum_{j=1}^N
\sum_{k=1}^{M_1^j} \left( \sum_{i=0}^{k-1} u_{ix}^j
(-\D_x)^{k-(i+1)} \right) 
\frac{\pl f}{\pl u_{kx}^j} \right)[\lambda \mf{u}]\,\frac{d \lambda}{\lambda}.
\end{equation}
The right hand side of (\ref{expressionF}) is identical to 
(\ref{1Dhomotopy}) with (\ref{1Dintegrand}).
\end{proof}
%
%
The homotopy operator (\ref{1Dhomotopy}) has been coded in {\it Mathematica} 
syntax and is part of the package {\tt HomotopyIntegrator.m} 
\cite{PooleHereman09software1} .
In extensive testing \cite{Poole09}, the expanded formula (\ref{1Dhomotopy})
out performed the implementation of the homotopy operator in 
\cite{HeremanAll05,HeremanIJQC06}, 
dramatically reducing CPU time on complicated expressions.
The code {\tt HomotopyOperator.m} also integrates differential expressions 
that {\it Mathematica's} {\courier Integrate} function fails to integrate.
For example, {\courier Integrate} cannot integrate the (exact) expression
\begin{equation}
\label{simplef}
f = u_x v_{2x} \cos u + v_{3x} \sin u - v_{4x}.
\end{equation}
The {\tt HomotopyIntegrator.m} code returns 
$F = \D_x^{-1} f = v_{2x} \sin u - v_{3x}.$ 
It is easy to verify that $\D_x F = f.$
Needless to say, one does not need the homotopy operator to integrate 
simple expressions like (\ref{simplef}) for it can easily be done with pen 
on paper. 
The homotopy operator is a tool for integrating long and more complicated 
differential functions, in particular, those where {\it Mathematica}'s 
{\courier Integrate} function fails and integration by hand becomes 
intractable or prone to errors. 
\vspace{-3mm}
\noindent
\section{Multi-Dimensional Homotopy Operators}
The 2\,D and 3\,D homotopy operators invert a divergence in 2\,D and 3\,D, 
respectively.
Such inversions are considerably more difficult to do by hand.
The 2\,D and 3\,D homotopy operators given in this paper 
were developed using the technique of proof in Theorem \ref{homotopy1Dtheorem}.
Therefore, the homotopy operators presented below are different from the ones 
in \cite{HeremanAll05,HeremanIJQC06,DeconinckNivala05,Cheviakov09}, 
where they were defined in terms of higher-order Euler operators.
Like in the 1\,D case, the expanded versions of multi-dimensional homotopy 
operators require considerably less computations \cite{Poole09}.
\begin{definition}
\label{twoDhomotopyoperator}
Let $f(x,y,\mf{u}^{(M)}(x,y))$ be an exact differential function involving 
two independent variables $\mf{x} = (x,y).$
The 2\,D homotopy operator is a ``vector" operator with two components,
\[
\left(\mc{H}_{\mf{u}(x,y)}^{(x)} f, \mc{H}_{\mf{u}(x,y)}^{(y)} f \right),
\]
where
\begin{equation}
\label{2Dhomotopy} 
\mc{H}_{\mf{u}(x,y)}^{(x)} f
\!=\! 
\int_{\lambda_0}^1 \left( \sum_{j=1}^{N} 
\mc{I}_{u^j(x,y)}^{(x)} f \right)
[\lambda {\mf{u}}]\,\frac{d \lambda}{\lambda}  
\,\;{\rm and}\;\,
\mc{H}_{\mf{u}(x,y)}^{(y)} f
\!=\! 
\int_{\lambda_0}^1 \left( 
\sum_{j=1}^{N} 
\mc{I}_{u^j(x,y)}^{(y)} f \right)[\lambda {\mf{u}}]\,\frac{d \lambda}{\lambda}.
\end{equation}
The $x$-integrand, $\mc{I}_{u^j(x,y)}^{(x)} f$, is given by 
\begin{eqnarray}
\label{integrand2Dx}
\!\!\!\!\!\!\!\!\!\!\mc{I}_{u^j(x,y)}^{(x)} f =
\sum_{k_1=1}^{M_1^j} \sum_{k_2=0}^{M_2^j}
\left( \sum_{i_1=0}^{k_1-1} \sum_{i_2=0}^{k_2}
B^{(x)}\, u_{i_1 x\,i_2 y}^j \left( -\D_x \right)^{k_1-i_1-1}
\!\left( -\D_y \right)^{k_2-i_2} \right) \frac{\pl f}{\pl u_{k_1 x\,k_2 y}^j},
\end{eqnarray}
with combinatorial coefficient $B^{(x)} = B(i_1,i_2,k_1,k_2)$ defined as
\[
B(i_1, i_2, k_1, k_2)  = \frac{{i_1 + i_2 \choose i_1}
{k_1 + k_2 - i_1 - i_2 - 1 \choose k_1 - i_1 - 1}}{{k_1 + k_2 \choose k_1}}.
\]
%
Similarly, the $y$-integrand, $\mc{I}_{u^j(x,y)}^{(y)} f,$ is defined as
\vspace*{-2mm}
\begin{eqnarray}
\label{integrand2Dy} \!\!\!\!\!\!\!\!\!\!
\mc{I}_{u^j(x,y)}^{(y)}f =
\sum_{k_1=0}^{M_1^j} \sum_{k_2=1}^{M_2^j}
\left( \sum_{i_1=0}^{k_1} \sum_{i_2=0}^{k_2-1} 
B^{(y)}\, u_{i_1 x\,i_2 y}^j \left( -\D_x \right)^{k_1-i_1}
\!\left( -\D_y \right)^{k_2-i_2-1} \right)\frac{\pl f}{\pl u_{k_1 x\,k_2 y}^j},
\end{eqnarray}
where $B^{(y)} = B(i_2,i_1,k_2,k_1).$
\end{definition}
%
\begin{definition}
\label{threeDhomotopyoperator}
Let $f( \mf{x}, \mf{u}^{(M)}(\mf{x}) )$ be a differential function of three 
independent variables 
where $\mf{x} = (x,y,z).$
The homotopy operator in 3\,D is a three-component (vector) operator,
%
\[
\left( \mc{H}_{\mf{u}(x,y,z)}^{(x)} f, \mc{H}_{\mf{u}(x,y,z)}^{(y)} f, 
\mc{H}_{\mf{u}(x,y,z)}^{(z)} f \right),
\]
where the $x$-component is given by
%
\[
\mc{H}_{\mf{u}(x,y,z)}^{(x)} f 
= \int_{\lambda_0}^1 
\left( \sum_{j=1}^{N} \mc{I}_{u^j(x,y,z)}^{(x)} f \right)
[\lambda {\mf{u}}]\,\frac{d \lambda}{\lambda}.
\]
The $y$- and $z$-components are defined analogously.
The $x$-integrand is given by 
\begin{eqnarray*}
\label{integrand3Dx}
I_{u^j(x,y,z)}^{(x)} f\!&=&\!\sum_{k_1=1}^{M_1^j}
\sum_{k_2=0}^{M_2^j} \sum_{k_3=0}^{M_3^j}
\sum_{i_1=0}^{k_1-1} \sum_{i_2=0}^{k_2} \sum_{i_3=0}^{k_3}
\left( B^{(x)}\, u_{i_1 x\,i_2 y\,i_3 z}^j \right. 
\nonumber \\ 
&& \ \ \ \ \ \ \ \ \ \ \left. \left( -\D_x \right)^{k_1-i_1-1}
\left( -\D_y \right)^{k_2-i_2} \left( -\D_z \right)^{k_3-i_3} \right)
\frac{\pl f}{\pl u_{k_1 x\,k_2 y\,k_3 z}^j}, 
\nonumber 
\end{eqnarray*}
with combinatorial coefficient $B^{(x)} = B(i_1, i_2, i_3, k_1, k_2, k_3)$ 
defined as
\[
B(i_1, i_2, i_3, k_1, k_2, k_3)
= \frac{{i_1 + i_2 + i_3 \choose i_1}\, {i_2 + i_3 \choose i_2}\,
{k_1 + k_2 + k_3 - i_1 - i_2 - i_3 - 1 \choose k_1 - i_1 - 1}\,
{k_2 + k_3 - i_2 - i_3 \choose k_2 - i_2}}
{{k_1 + k_2 + k_3 \choose k_1}\,{k_2 + k_3 \choose k_2}}.
\]
%
As one might expect, the integrands $I_{u^j(x,y,z)}^{(y)} f$ and
$I_{u^j(x,y,z)}^{(z)} f $ are defined analogously. 
Based on cyclic permutations, they have combinatorial coefficients 
$B^{(y)} = B(i_2,i_3,i_1,k_2,k_3,k_1)$ and 
$B^{(z)} = B(i_3,i_1,i_2,k_3,k_1,k_2)$, respectively.
\end{definition}
The homotopy with $\lambda_0 = 0$ is used, except when singularities 
at $\lambda = 0$ occur.
\begin{theorem}
\label{homotopy2D3Dtheorem}
Let $f = f(\mf{x}, \mf{u}^{(M)}(\mf{x}))$ be exact, i.e.\ 
$f = \mathrm{Div}\,\mf{F}$ for some $\mf{F}(\mf{x}, \mf{u}^{(M-1)}(\mf{x})).$
Then, in the 2\,D case, 
$\mf{F} = \mathrm{Div}^{-1} f = 
\left(\mc{H}_{\mf{u}(x,y)}^{(x)} f, \mc{H}_{\mf{u}(x,y)}^{(y)} f \right).$
Analogously, in 3\,D, 
$\mf{F} = \mathrm{Div}^{-1} f = 
\left( \mc{H}_{\mf{u}(x,y,z)}^{(x)} f, \mc{H}_{\mf{u}(x,y,z)}^{(y)} f, 
\mc{H}_{\mf{u}(x,y,z)}^{(z)} f \right).$
\end{theorem}
\begin{proof}
The very lengthy proof is similar to that of Theorem \ref{homotopy1Dtheorem}.
It does not add much to the understanding of the theory.
Nevertheless, details can be found in \cite[Appendix A]{Poole09}.
\end{proof}
%
The following 2\,D example demonstrates how (\ref{2Dhomotopy}) works.
\begin{example}
\label{2Dhomotopyexample}
With $\mf{u} = (u,v)$ and $\mf{x} = (x,y),$ let
\begin{eqnarray}
\label{2Dexampledivergence} 
f( \mf{x},\mf{u}^{(3)}(\mf{x}) ) 
&=& u^2 u_{x2y} + 2 u u_x u_{2y} + 3 u_x v_x \cos v 
- 4 u_y v_x v_{xy} - 2 u_{2y} v_x^2 
\nonumber \\ 
&& + v_{2y} \cos u + 3 u_{2x} \sin v - u_y v_y \sin u.
\end{eqnarray}
Using (\ref{integrand2Dx}), compute 
\begin{eqnarray*}
\label{xintegrandsexample}
I_{u(x,y)}^{(x)} 
&=& u \frac{\pl f}{\pl u_x} + \left(u_x \Id - u \D_x \right) 
\frac{\pl f}{\pl u_{2x}} 
+ \frac{1}{3} \left( u_{2y} \Id - u_y \D_y + u \D_y^2 \right) 
\frac{\pl f}{\pl u_{x2y}} 
\nonumber \\
&=& 3 u^2 u_{2y} + 3 u_x \sin v, 
\nonumber \\ 
I_{v(x,y)}^{(x)} 
&=& v \frac{\pl f}{\pl v_x} + \frac{1}{2} \left(v_y \Id - v \D_y \right) 
\frac{\pl f}{\pl v_{xy}} 
\nonumber \\
&=& 3 u_x v \cos v - 2 u_y v v_{xy} -2 u_y v_x v_y - 2 u_{2y} v v_x. 
\end{eqnarray*}
Likewise, using (\ref{integrand2Dy}), compute 
\begin{eqnarray*}
\label{yintegrandsexample}
I_{u(x,y)}^{(y)} 
&=& u \frac{\pl f}{\pl u_y} + \left(u_y \Id - u \D_y \right) 
\frac{\pl f}{\pl u_{2y}}
+ \frac{1}{3} \left(2 u_{xy} \Id - u_y \D_x - u_x \D_y + 2 u \D_x \D_y \right) 
\frac{\pl f}{\pl u_{x2y}} 
\nonumber \\ 
&=& -u v_y \sin u - 2 u_y v_x^2, 
\nonumber \\ 
I_{v(x,y)}^{(y)} 
&=& v \frac{\pl f}{\pl v_y} + \left(v_y \Id - v \D_y \right) 
\frac{\pl f}{\pl v_{2y}}
+ \frac{1}{2} \left( v_x \Id - v \D_x \right) \frac{\pl f}{\pl v_{xy}} 
\nonumber \\ 
&=& 2 u_y v v_{2x} - 2 u_y v_x^2 + 2 u_{xy} v v_x + v_y \cos u.
\end{eqnarray*}
Then, using (\ref{2Dhomotopy}), compute
\begin{eqnarray*}
\label{integrand2Dxexample}
\mc{H}_{\mf{u}(x,y)}^{(x)} f
&=& \int_0^1 \left( 
\mc{I}_{u(x,y)}^{(x)} f + \mc{I}_{v(x,y)}^{(x)} f \right)
[\lambda \mf{u}]\,\frac{d \lambda}{\lambda}
\nonumber \\ 
&=& \int_0^1 \left( 3 \lambda^2 u^2 u_{2y} + 3 u_x \sin \lambda v
+ 3 \lambda u_x v \cos \lambda v - 2 \lambda^2 u_y v v_{xy} \right.
\nonumber \\ 
&& \left. 
- 2 \lambda^2 u_y v_x v_y - 2 \lambda^2 u_{2y} v v_x \right)\, d \lambda
\nonumber \\ 
&=& u^2 u_{2y} + 3 u_x \sin v - \tfrac{2}{3} u_y v v_{xy}
- \tfrac{2}{3} u_y v_x v_y - \tfrac{2}{3} u_{2y} v v_x, 
\nonumber \\ 
\mc{H}_{\mf{u}(x,y)}^{(y)} f 
&=& \int_0^1 \left( 
\mc{I}_{u(x,y)}^{(y)} f + \mc{I}_{v(x,y)}^{(y)} f \right)
[\lambda \mf{u}]\,\frac{d \lambda}{\lambda}
\nonumber \\ 
&=& \int_0^1 \!\! \left( v_y \cos \lambda u
- \lambda u v_y \sin \lambda u - 4 \lambda^2 u_y v_x^2
+ 2 \lambda^2 u_y v v_{2x} + 2 \lambda^2 u_{xy} v v_x \right)\,d \lambda
\nonumber \\ 
&=& \tfrac{2}{3} u_y v v_{2x} + \tfrac{2}{3} u_{xy} v v_x
- \tfrac{4}{3} u_y v_x^2 + v_y \cos u.
\end{eqnarray*}
Thus, the homotopy operator gives the vector
\begin{equation}
\label{homotopyresult2Dexample}
\mf{F} = \mathrm{Div}^{-1} f 
= \left( \begin{array}{c}
u^2 u_{2y} + 3 u_x \sin v - \frac{2}{3} u_y v v_{xy}
- \frac{2}{3} u_y v_x v_y - \frac{2}{3} u_{2y} v v_x \\ 
\frac{2}{3} u_y v v_{2x} + \frac{2}{3} u_{xy} v v_x
- \frac{4}{3} u_y v_x^2 + v_y \cos u 
\end{array} \right).
\end{equation}
\end{example}
Obviously, there are infinitely many choices for $\mf{F}$ because the 
addition to $\mf{F}$ of a 2\,D ``curl vector,'' 
$\mf{K} =  (\D_y \theta, - \D_x \theta),$ where $\theta$ 
is an arbitrary differential expression, will produce an identical divergence.
Indeed, 
$\mathrm{Div}\,\mf{G} = \mathrm{Div}\,(\mf{F} + \mf{K}) 
= \mathrm{Div}\,\mf{F} + \D_x(\D_y \theta) - \D_y(\D_x \theta) = 
\mathrm{Div}\,\mf{F} = f.$
The same happens in 3\,D where a curl vector,
$\mf{K} = (\D_y \eta - \D_z \xi, 
\D_z \theta - \D_x \eta, \D_x \xi - \D_y \theta)$ involves three arbitrary 
differential functions $\theta, \eta$ and $\xi.$
Often a concise result can be obtained for (\ref{homotopyresult2Dexample}) 
by removing undesirable curl terms.
This is the subject of the next section.
\vspace{-3mm}
\noindent
\section{Removing Curl Terms}
\label{secremovecurl}
The homotopy operator acting on an exact multi-dimensional differential
function returns an $\mf{F}$ which includes a curl vector.
Although the homotopy operator consistently leads to a particular choice for
the curl vector, $\mf{F}$ can be of unmanageable size due to the presence of 
dozens, if not hundreds, of unwanted terms.
Removing the curl terms makes $\mf{F}$ shorter.

There are several methods for removing curl vectors.
Doing so by hand requires educated guess work. 
Ideally, one could modify the integrand of the homotopy operator so that 
the least number of curl terms would be generated. 
We are currently investigating this strategy.
For now, we propose a simple, yet effective algorithm that uses only linear 
algebra.

To remove the curl terms from (\ref{homotopyresult2Dexample}), first 
attach undetermined coefficients, $k_1, \dots, k_p$, to the terms of 
(\ref{homotopyresult2Dexample}) and call the new vector $\mf{\tilde{F}}.$
Doing so, 
\begin{equation}
\label{veccurlconstants}
\mf{\tilde{F}}= \left( \begin{array}{c}
k_1 u^2 u_{2y} + k_2 u_x \sin v + k_3 u_y v v_{xy}
+ k_4 u_y v_x v_y +  k_5 u_{2y} v v_x \\ 
k_6 u_y v v_{2x} + k_7 u_{xy} v v_x
+ k_8 u_y v_x^2 + k_9 v_y \cos u 
\end{array} \right).
\end{equation}
Next, compute 
\begin{eqnarray}
\label{divwithconstants}
\mathrm{Div}\,\mf{\tilde{F}} 
&=& 2 k_1 u u_x u_{2y} + k_1 u^2 u_{x2y} + k_2 u_x v_x \cos v 
+ k_2 u_{2x} \sin v 
\nonumber \\ 
&& {} \!+\! (k_3 + k_4 + 2 k_8) u_y v_x v_{xy} + (k_3 + k_6) u_y v v_{2xy}
+ (k_3 + k_7) u_{xy} v v_{xy} 
\nonumber \\ 
&& {} \!+\! (k_4 + k_6) u_y v_{2x} v_y + (k_4 + k_7) u_{xy} v_x v_y
+ (k_5 + k_6) u_{2y} v v_{2x} 
\nonumber \\ 
&& {} \!+\! (k_5 + k_7) u_{x2y} v v_x + (k_5 + k_8) u_{2y} v_x^2
+ k_9 v_{2y} \cos u - k_9 u_y v_y \sin u.
\end{eqnarray}
Since $\mf{F}$ in (\ref{homotopyresult2Dexample}) and $\tilde{\mf{F}}$ 
in (\ref{veccurlconstants}) should differ only by a curl vector, 
their divergences must be identical.  
Thus, 
$\mathrm{Div}\,\mf{\tilde{F}} = \mathrm{Div}\,\mf{F} \equiv f$ in 
(\ref{2Dexampledivergence}).
Gathering like terms in (\ref{2Dexampledivergence}) and 
(\ref{divwithconstants}) leads to the linear system
\begin{equation}
\label{systemforexample}
\begin{array}{rcrrcrrcrrcr}
k_1 & = & 1, & \qquad \quad k_2 & = & 3, 
& \quad k_3 + k_4 + 2k_8 & = & -4, & \quad k_3 + k_6 & = & 0,
\\ 
k_3 + k_7 & = & 0, & k_4 + k_6 & = & 0,
& k_4 + k_7 & = & 0, & k_5 + k_6 & = & 0,
\\ 
k_5 + k_7 & = & 0, & k_5 + k_8 & = & -2, & k_9 & = & 1. & & &
\end{array}
\end{equation}
%
To find the undetermined coefficients that will produce a $\tilde{\mf{F}}$ 
without curl terms, first solve for the $k_i$ that appear only in equations 
with non-zero right-hand sides.
In (\ref{systemforexample}), clearly $k_1 = 1$, $k_2 = 3$, and $k_9 = 1.$
The next variable to solve is $k_8$ since it appears in equations with
non-zero right hand sides.
Coefficients $k_3$ through $k_7$ appear in equations with zero right-hand 
sides; these equations are solved last.
Solving (\ref{systemforexample}) in this order yields
\begin{equation}
\label{solutionexample}
\begin{array}{rcrrcrrclrcr}
k_1 & = & 1, & \qquad k_2 & = & 3, & \qquad k_3 & = & -k_7, 
& \qquad k_4 & = & -k_7, 
\\ 
k_5 & = & -k_7, & k_6 & = & k_7, & k_8 & = & k_7 - 2, & k_9 & = & 1,
\end{array}
\end{equation}
where $k_7$ is arbitrary. 
Set $k_7 = 0$ to eliminate as many terms as possible in 
(\ref{veccurlconstants}).
Substitute (\ref{solutionexample}) with $k_7 = 0$ into 
(\ref{veccurlconstants}), to get
\begin{equation}
\label{finalFtilde}
\mf{\tilde{F}} = \left( \begin{array}{c}
u^2 u_{2y} + 3 u_x \sin v \\ 
v_y \cos u - 2 u_y v_x^2 
\end{array} \right).
\end{equation}
By construction, 
$\mathrm{Div}\,\mf{\tilde{F}} = \mathrm{Div}\,\mf{F} = f,$ but 
(\ref{finalFtilde}) is free of curl terms.
\vspace{-3mm}
\noindent
\section{Extending the Applicability of the Homotopy Operator}
\label{seckerM}
It follows from (\ref{mf}), that the 1\,D homotopy operator in 
Definition \ref{oneDhomotopyoperator}, will not work when 
$f(x,\mf{u}^{(M)}(x))$ (or for that matter any term of it) belongs 
to $\mathrm{Ker}\,\M.$
Indeed, if $\M f = 0$, in particular when $f$ satisfies the conditions of
Theorem \ref{TheoremkernelM}, then 
$\sum_{j=1}^N \mc{I}_{u^j(x)} f = C$, where $C$ is an arbitrary constant.
Thus, the homotopy operator would return $C,$ in fact $0$ for it ignores 
constants of integration.

The next two examples (in 1\,D) demonstrate a problem that occurs when 
some terms of $f$ belong to $\mathrm{Ker}\,\M$ while others do not.
\begin{example}
\label{degree0example1D}
Let $\mf{u} = (u^1,u^2) = (u,v)$ and $F = \frac{u_x}{v}.$
Then,
\[
f = \D_x F = \frac{u_{2x} v - u_x v_x}{v^2} 
  = \frac{u_{2x}}{v} - \frac{u_x v_x}{v^2}.
\]
Obviously, $f$ is of degree zero.
Hence, $\M f = 0$ by Theorem \ref{TheoremkernelM}.
Therefore, trying to integrate $f$ with the homotopy operator will fail.
Indeed, using (\ref{1Dintegrand}) for $u(x)$ and $v(x)$, respectively, yields
\begin{eqnarray*}
\label{integrandsexamplekernel1}
I_{u(x)} f 
&=& u \frac{\pl f}{\pl u_x} + (u_x \Id - u \D_x) \frac{\pl f}{\pl u_{2x}}
= - u \frac{v_x}{v^2} + u \frac{v_x}{v^2} + \frac{u_x}{v}
= \frac{u_x}{v}, 
\nonumber \\
I_{v(x)} f 
&=& v \frac{\pl f}{\pl v_x} 
= - v \frac{u_x}{v^2} 
= - \frac{u_x}{v}.
\nonumber
\end{eqnarray*}
Clearly, $I_{u(x)} f + I_{v(x)} f = 0$ and (\ref{1Dhomotopy}) gives 
$0$ instead of $F.$
\end{example}
\noindent 
Furthermore, if $f$ is the sum of monomial, rational, or irrational terms
(and fractions thereof) where any of these terms are in $\mathrm{Ker}\,\M,$
then the homotopy operator will return an incorrect result. 
Indeed, the homotopy operator would correctly integrate the terms not in 
$\mathrm{Ker}\,\M,$ but annihilate the terms in $\mathrm{Ker}\,\M.$  
The following example highlights this situation.
\begin{example}
\label{degree0examplemultiterm}
Let $\mf{u} = (u^1,u^2) = (u,v)$ and $F = \frac{u^2 + v}{u - v}.$  
Then,
\begin{eqnarray}
\label{exampleftroubleterms}
f = \D_x F &=& 
\frac{u^2 u_x + u^2 v_x - 2 u u_x v + u v_x - u_x v}{(u - v)^2}
\nonumber \\
&=& \frac{u^2 u_x}{(u - v)^2} + \frac{u^2 v_x}{(u - v)^2} 
-2\frac{u u_x v}{(u - v)^2} + \frac{u v_x}{(u - v)^2} -\frac{u_x v}{(u - v)^2}
\end{eqnarray}
is exact by construction. 
Applying the degree operator to $f,$ term by term, gives
\begin{displaymath}
\label{Mtoexampleftroubleterms}
\M f = \frac{u^2 u_x}{(u - v)^2} + \frac{u^2 v_x}{(u - v)^2}
       - 2 \frac{u v u_x}{(u - v)^2} + 0 + 0.
\nonumber
\end{displaymath}
Because the last two terms of (\ref{exampleftroubleterms}) are annihilated 
by $\M,$ the homotopy operator will incorrectly integrate $f.$  
One would obtain $F = \frac{u^2}{u - v}$ instead of 
$F = \frac{u^2 + v}{u - v}.$  
\end{example}
A method to overcome these shortcomings will now be proposed.
A coordinate for the ``jet'' space where $f(\mf{x},\mf{u}^{(M)}(\mf{x}))$
resides is given by
\[
\mf{u}^{(M)}(\mf{x}) =
(u^1,u_x^1,u_y^1,u_z^1,u_{2x}^1,u_{2y}^1,u_{2z}^1,u_{xy}^1,
\dots, u^2, u_x^2, \dots, u^N_{M_1^N x\, M_2^N y\, M_3^N z}), 
\]
which has origin $\mf{0} = (0, \dots, 0).$  
By shifting the coordinate away from $\mf{0}$, 
$f$ can be taken out of $\mathrm{Ker}\,\M.$
It suffices to shift one of the variables that appear in the denominator 
of the term in $\mathrm{Ker}\,\M$.
After integrating, the shift must be undone to put the integral at the origin.
The following examples illustrate the procedure.
\begin{example}
Returning to Example \ref{degree0example1D}, where 
$f(x, \mf{u}^{(1)}(x)) = \frac{u_{2x} v - u_x v_x}{v^2}.$ 
Since $v$ appears in the denominator, replace $v$ with $v - v_0$ to get
$f_0 = \frac{u_{2x}(v - v_0) - u_x v_x}{(v - v_0)^2}.$
Then use the homotopy operator to integrate $f_0.$
Using (\ref{1Dintegrand}), first compute 
\begin{displaymath}
\mc{I}_{u(x)} f_0 = \frac{u_x}{v-v_0} \qquad 
\mbox{\rm and} \qquad
\mc{I}_{v(x)} f_0 = -\frac{u_x v}{(v-v_0)^2}.
\nonumber
\end{displaymath}
Next, using (\ref{1Dhomotopy}), compute
\begin{eqnarray*}
\label{homotopydegree0example1D}
\mc{H}_{\mf{u}(x)} f_0  
&=& \int_0^1 \left( \frac{u_x}{v-v_0} - \frac{u_x v}{(v-v_0)^2} \right)
[\lambda \mf{u}]\,\frac{d \lambda}{\lambda}
= - \int_0^1 \frac{u_x v_0}{(\lambda v - v_0)^2}\,d \lambda 
\nonumber \\ 
&=& \left. { \frac{u_x v_0}{v (\lambda v - v_0)} }\right|_{0}^{1}
= \frac{u_x}{v - v_0}.
\nonumber
\end{eqnarray*}
Finally, remove the shift by replacing $v$ with $v + v_0$ 
(or simply set $v_0 = 0)$ to get
\[
F = \mc{H}_{\mf{u}(x)} f = \frac{u_x}{v}.
\]
\end{example}
\begin{example}
To integrate $f$ from Example \ref{degree0examplemultiterm}, 
the shift only needs to be applied to the last two terms of
(\ref{exampleftroubleterms}) because these are in $\mathrm{Ker}\,\M.$
Therefore, we will apply the homotopy operator to
\[
f_0 = \frac{u^2 u_x}{(u - v)^2} + \frac{u^2 v_x}{(u - v)^2}
- 2 \frac{u u_x v}{(u - v)^2} + \frac{(u -u_0) v_x}{(u - u_0 - v)^2} 
- \frac{u_x v}{(u - u_0 - v)^2}, 
\]
where $u$ has been replaced by $u - u_0$ in the last two terms.
Note that it suffices to shift just one of the two variables in the 
denominator.
Keeping the number of shifts to a minimum leads to a simpler integrand 
for the integration over $\lambda.$
The integrands for the homotopy operator are
\begin{eqnarray*}
\mc{I}_{u(x)} f_0 
&=& \frac{u^2 (u - 2 v)}{(u - v)^2} - \frac{uv}{(u - u_0 - v)^2} 
\quad \mbox{{\rm and}} \quad 
\mc{I}_{v(x)} f_0 
= \frac{u^2 v}{(u - v)^2} + \frac{v(u - u_0)}{(u - u_0 - v)^2}.
\nonumber
\end{eqnarray*}
The first terms in each integrand have not been shifted. 
However, the second terms have been shifted.
Applying (\ref{1Dhomotopy}),
\begin{eqnarray*}
\mc{H}_{\mf{u}(x)}^{(x)} f_0 
&=& \int_0^1 \left( \frac{u^2}{u - v}
- \frac{u_0 v}{(\lambda u - u_0 - \lambda v)^2} \right)\,d \lambda 
\nonumber \\ 
&=& \left. \Big( \frac{\lambda u^2}{u - v}
+ \frac{u_0 v}{(u - v) (\lambda u - u_0 - \lambda v)} \Big)\right|_{0}^{1}
\nonumber \\ 
&=& \frac{u^2}{u - v} + \left( \frac{u_0 v}{(u - v) ( u - u_0 - v)}
+ \frac{v}{(u - v)} \right)
\nonumber \\ 
&=& \frac{u^2}{u - v} + \frac{v}{u - u_0 - v}.
\nonumber
\end{eqnarray*}
Finally, setting $u_0 = 0$ yields the desired result, 
$F = \frac{u^2 + v}{u - v}.$
\end{example}
Similar difficulties might occur with functions in fractional form 
involving multiple independent variables.
Again, a shift of a coordinate in the denominator will solve the problem.
However, the complexity of the integration over $\lambda$ increases rapidly 
when there are multiple independent variables (with shifts).
Implementation of the 2\,D and 3\,D homotopy operators (in Definitions 
\ref{twoDhomotopyoperator} and \ref{threeDhomotopyoperator})
requires the same amount of care.
A discussion of the kernel of the homotopy operator, including an alternate 
strategy for dealing with the integration of homogeneous expressions can be 
found in \cite{DeconinckNivala08}. 
\vspace{-3mm}
\noindent
\section{Application: Computation of Conservation Laws of Nonlinear PDEs}
This Section covers the application of the homotopy operator to the 
computation of conservation laws of nonlinear PDEs.
We developed a {\it Mathematica} package, {\tt ConservationLawsMD.m} 
\cite{PooleHereman09software2}, that automates these computations for 
nonlinear polynomial PDEs with a maximum of three space variables 
in addition to time.
\begin{definition}
A conservation law for a given PDE with independent variables $t$ and $\mf{x}$
is defined as
\begin{equation}
\label{genconlaw}
\D_t \rho + \mathrm{Div}\,\mf{J} = 0,
\end{equation}
where $\rho = \rho(t, \mf{x}, \mf{u}^{(M)}(t,\mf{x}))$ is the 
{\it conserved density}, 
$\mf{J} = \mf{J} (t, \mf{x}, \mf{u}^{(P)}(t,\mf{x}))$ is the
{\it associated flux}. 
$\D_t$ is the total derivative with respect to $t$ and $\mathrm{Div}$
is the total divergence with respect to $\mf{x}.$ 
(\ref{genconlaw}) is satisfied on solutions of the given PDE 
\cite{AblowitzClarkson91}.
\end{definition}
A conservation law is found by first computing the density, $\rho,$
followed by the computation of the flux $\mf{J}.$
The latter will require the use of the homotopy operator.

Following the approach by Hereman {\it et al}.\ 
\cite{HeremanAll09,HeremanAll05,HeremanIJQC06}, 
a candidate density is built as a linear combination 
(with undetermined coefficients) of differential terms which are invariant 
under the scaling symmetry of the given PDE.
Once the form of $\rho$ is determined, one computes $\D_t \rho$ and removes 
all time-derivatives using the PDE.
By (\ref{genconlaw}), $\D_t \rho$ must be a divergence.
Thus, using Theorem \ref{zeroeulerexact}, one requires that 
%
\[
\mc{L}_{u^j(\mf{x})} (\D_t \rho) \equiv 0, \quad j = 1, \dots, N,
\]
where $\mc{L}_{u^j(\mf{x})}$ is the Euler operator.
This leads to a linear system for the undetermined coefficients. 
Substituting its solution into the candidate for $\rho$ gives the actual
density. 
Finally, the flux $\mathbf{J} = - \mathrm{Div}^{-1} (\D_t\rho)$ is computed 
with the homotopy operator.

To illustrate the method, we compute a conservation law for the 
$(2+1)-$ dimensional Zakharov-Kuznetsov (ZK) equation 
\cite{ZakharovKuznetsov74} which describes ion-acoustic solitons in 
magnetized plasmas.
After re-scaling, the PDE takes the form
\begin{equation}
\label{ZKequation}
u_t + \alpha u u_x + \beta (\Delta u)_x = 0,
\end{equation}
where $\alpha$ and $\beta$ are parameters and 
$\Delta = \frac{\partial^2}{\partial x^2} + \frac{\partial^2}{\partial y^2}$ 
is the Laplacian.
\vspace{-3mm}
\noindent
\subsection{Computation of the Density}
It is easy to verify that (\ref{ZKequation}) is invariant under the 
scaling (dilation) symmetry,
\begin{equation}
\label{ZKscalingsymmetry}
(t,x,y,u) \rightarrow (\lambda^{-3} t, \lambda^{-1} x,
\lambda^{-1} y, \lambda^2 u).
\end{equation}
This scaling symmetry can be computed as follows. 
Assume that (\ref{ZKequation}) scales uniformly under 
\begin{equation}
\label{ZKscalingsymmetrygeneral}
(t,x,y,u) \rightarrow (T,X,Y,U) = (\lambda^{a} t, \lambda^{b} x,
\lambda^{c} y, \lambda^{d} u),
\end{equation}
with undetermined (rational) exponents $a,b,c,$ and $d.$
Applying the chain rule to (\ref{ZKequation}) yields
\begin{eqnarray}
\label{ZKequationnewvars}
&& \lambda^{a-d} \, U_T 
+ \alpha \lambda^{b-2d} \, U U_X 
+ \beta \lambda^{3b-d} \, U_{3X} 
+ \beta \lambda^{b+2c-d} \, U_{X2Y}  
\nonumber \\
&& = \lambda^{a-d} 
( U_T + \alpha \lambda^{b-d-a} \, U U_X + \beta \lambda^{3b-a} \, U_{3X} 
+ \beta \lambda^{b+2c-a} \, U_{X2Y} ) = 0.
\end{eqnarray}
One obtains (\ref{ZKequation}) in the new variables $(T,X,Y,U),$
up to the common factor $\lambda^{a-d},$ if 
$b - d - a = 3 b - a = b + 2 c - a = 0.$
Setting $b = -1,$ one gets $a = -3, c = -1,$ and $d = 2,$
which determines (\ref{ZKscalingsymmetry}).
A more algorithmic method for computing scaling symmetries can be found in
\cite{HeremanAll09,HeremanAll05,HeremanIJQC06}.

The scaling symmetry carries over to conservation laws and therefore can 
be used to construct densities \cite{KruskalAll70}.
Indeed, since (\ref{genconlaw}) holds on solutions of (\ref{ZKequation}), 
both $\rho$ and $\mf{J}$ must obey the scaling symmetry 
(\ref{ZKscalingsymmetry}) of given PDE.
For example, (\ref{ZKequation}) is a conservation law itself, with
\begin{equation}
\label{ZKrank2conlaw}
\rho^{(1)} = u,\quad \mbox \quad 
\mf{J}^{(1)} = (\tfrac{1}{2} \alpha u^2 + \beta u_{2x},\,\, \beta u_{xy}),
\end{equation}
expressing conservation of mass.
It is straightforward to verify that 
\begin{equation}
\label{ZKrank4conlaw}
\rho^{(2)} = u^2,
\; \mbox \;
\mf{J}^{(2)} = 
(\tfrac{2 \alpha}{3} u^3 - \beta (u_x^2 - u_y^2) + 2 \beta u (u_{2x} + u_{2y}),
\,\, - 2 \beta u_x u_y ).
\end{equation}
also are a density-flux pair of (\ref{ZKequation});
the corresponding conservation law expresses conservation of momentum.
Note that the densities in (\ref{ZKrank2conlaw}) and 
(\ref{ZKrank4conlaw}) scale with $\lambda^2$ and $\lambda^4,$ respectively. 
For brevity, we say that $\rho^{(1)}$ has {\it rank} $2;$ 
$\rho^{(2)}$ has rank $4.$
The fluxes $\mf{J}^{(1)}$ and $\mf{J}^{(2)}$ have ranks $4$ and $6.$
As a whole, the conservation laws (\ref{genconlaw}) with 
$(\rho^{(1)},\mf{J}^{(1)})$ and $(\rho^{(2)},\mf{J}^{(2)})$ have ranks 
$5$ and $7,$ respectively.

We will construct a third density, $\rho^{(3)},$ which has rank $6,$
i.e.\ each monomial in $\rho^{(3)}$ depends on $u$ and its derivatives 
so that every term scales with $\lambda^6.$
First, construct the list $\mc{P} = \{u^3,u^2,u\}$ containing powers of the 
dependent variable of rank 6 or less.  
Second, bring all terms in $\mc{P}$ with a rank less that 6 up to rank 6 by 
applying $\D_x$ and $\D_y.$
Considering all possible combinations of derivations yields
\begin{eqnarray*}
\label{ZKrank6qlist}
\mc{Q} &=& \{u^3, u_x^2, u u_{2x}, u_y^2, u u_{2y},
u_x u_y, u u_{xy}, u_{4x}, u_{3xy}, u_{2x2y}, u_{x3y}, u_{4y}\}.
\nonumber
\end{eqnarray*}
Third, remove all terms from $\mc{Q}$ that are divergences (because they 
belong to the flux).
For example, $u_{4x}$ = $\mathrm{Div}\,(u_{3x},0)$.
This leaves
\begin{eqnarray*}
\mc{Q} &=& \{u^3, u_x^2, u u_{2x}, u_y^2, u u_{2y}, u_x u_y, u u_{xy}\}.
\end{eqnarray*}
Fourth, find terms that are divergence-equivalent and remove all but the 
lowest order term.
Two or more terms are divergence-equivalent when a linear combination
of the terms is a divergence.  
For example, $u_x^2$ and $u u_{2x}$ are divergence-equivalent 
since $u_x^2 + u u_{2x} = \mathrm{Div}\,(u u_x,0).$ 
Consequently, $u u_{2x}$ can be removed from $\mc{Q}.$
Doing so, $\mc{Q} = \{u^3, u_x^2, u_y^2, u_x u_y\}.$
If divergences and divergence-equivalent terms were not removed, 
they would lead to trivial and equivalent conservation laws, respectively.

Now, form a candidate density by linearly combining the terms in $\mc{Q}$ 
with undetermined coefficients $c_i,$
\begin{equation}
\label{ZKrank6candidate}
\rho^{(3)} = c_1 u^3 + c_2 u_x^2 + c_3 u_y^2 + c_4 u_x u_y.
\end{equation}
All, part, or none of the candidate density may be an actual density for 
(\ref{ZKequation}).
Indeed, the candidate density may be trivial or a linear combination of two 
or more independent densities.
Computation of the undetermined coefficients will reveal the nature of 
the candidate density.

With (\ref{ZKrank6candidate}), compute 
\begin{equation}
\label{ZKDtrhorank6}
\D_t \rho^{(3)} = 3 c_1 u^2 u_t + 2 c_2 u_x u_{tx} + 2 c_3 u_y u_{ty} 
+ c_4 ( u_y u_{tx} + u_x u_{ty} ).
\end{equation}
Using (\ref{ZKequation}), replace $u_t$ with 
$-(\alpha u u_x + \beta u_{3x} + \beta u_{x2y}).$
Set $E = -\D_t \rho^{(3)},$ to get
\begin{eqnarray}
\label{ZKreplDtrhorank6}
E &=& 3 c_1 u^2 ( \alpha u u_x + \beta u_{3x} + \beta u_{x2y} )
+ 2 c_2 u_x ( \alpha u u_x  + \beta u_{3x} + \beta u_{x2y} )_x 
\nonumber \\
&& {} + 2 c_3 u_y ( \alpha u u_x  + \beta u_{3x} + \beta u_{x2y} )_y 
+ c_4 u_y (\alpha u u_x  + \beta u_{3x} + \beta u_{x2y} )_x 
\nonumber \\
&& {} + c_4 u_x (\alpha u u_x  + \beta u_{3x} + \beta u_{x2y} )_y.
\end{eqnarray}
Since (\ref{ZKreplDtrhorank6}) must be a divergence, 
use (\ref{zeroeulerscalarux2D}) and require
\begin{eqnarray}
\label{zeroEulerDtrhorank6} 
\mc{L}_{u(x,y)} E 
&=& -2 \big( (3 c_1  \beta  + c_3 \alpha) u_x u_{2y}
+ 2 (3 c_1 \beta + c_3 \alpha) u_y u_{xy}
+ 3 (3 c_1 \beta + c_2 \alpha) u_x u_{2x} \big) 
\nonumber \\ 
&& {} + 2 c_4 \alpha u_x u_{xy} + c_4 \alpha u_y u_{2x} \equiv 0.
\end{eqnarray}
It follows from (\ref{zeroEulerDtrhorank6}) that
\begin{equation}
\label{ZKrank6undetsystem}
3 c_1 \beta + c_2 \alpha = 0, \qquad 3 c_1 \beta + c_3 \alpha = 0,
\qquad \alpha c_4  = 0.
\end{equation}
Before solving the system it is necessary to check for potential compatibility
conditions on the parameters $\alpha$ and $\beta.$
This is done by setting each $c_i = 1$, one at a time, and algebraically
eliminating the other undetermined coefficients.
See, e.g., \cite{Poole09,GoktasHeremanJSC97} for details about computing 
compatibility conditions.
It turns out that there are no compatibility conditions for 
(\ref{ZKrank6undetsystem}) and the solution is
\begin{equation}
\label{ZKrank6undetsoln}
c_2 = - \tfrac{3 \beta}{\alpha} c_1, \qquad 
c_3 = -\tfrac{3 \beta}{\alpha} c_1, \qquad 
c_4 = 0,
\end{equation}
where $c_1$ is arbitrary.  
Substitute (\ref{ZKrank6undetsoln}) into (\ref{ZKrank6candidate}) and set 
$c_1 = 1$ to get
\begin{equation}
\label{ZKrank6density}
\rho^{(3)} = u^3 - \tfrac{3\beta}{\alpha} (u_x^2 + u_y^2).
\end{equation}
\vspace{-5mm}
\noindent
\subsection{Computation of the Flux}
After substitution of (\ref{ZKrank6undetsoln}) and $c_1 = 1$ into 
(\ref{ZKreplDtrhorank6}),
\begin{eqnarray*}
\label{Ewithsolns}
E &=& 3 u^2 (\alpha u u_x  + \beta u_{3x} + \beta u_{x2y} )
- \tfrac{6 \beta}{\alpha} u_x 
(\alpha u u_x  + \beta u_{3x} + \beta u_{x2y})_x 
\nonumber \\
&& {} - \tfrac{6 \beta}{\alpha} u_y 
(\alpha u u_x  + \beta u_{3x} + \beta u_{x2y} )_y.
\nonumber
\end{eqnarray*}
Since $E = \mathrm{Div}\,\mf{J}^{(3)},$ the flux $\mf{J}^{(3)}$ can be 
computed with the 2\,D homotopy operator which inverts divergences.
Using (\ref{2Dhomotopy}), the integrands are
\begin{eqnarray*}
I_{u(x,y)}^{(x)} E 
&=& 3 \alpha u^4 + \beta \left( 9 u^2 (u_{2x} + \tfrac{2}{3} u_{2y})
- 6  u (3 u_x^2 + u_y^2) \right) 
+ \tfrac{\beta^2}{\alpha} \left( 6 u_{2x}^2 + 5 u_{xy}^2
+ \tfrac{3}{2} u_{2y}^2 \right. 
\nonumber \\ 
&& \left. {} + \tfrac{3}{2} u ( u_{4y} + u_{2x2y})
- u_x (12 u_{3x} + 7 u_{x2y}) - u_y (3 u_{3y} + 8 u_{2xy})
+ \tfrac{5}{2} u_{2x} u_{2y} \right), 
\nonumber 
\end{eqnarray*}
and
\begin{eqnarray*}
I_{u(x,y)}^{(y)} E &=& 3 \beta  u ( u u_{xy} - 4 u_x u_y)
- \tfrac{\beta^2}{2 \alpha} \left( 3 u ( u_{3xy} + u_{x3y})
+ u_x (13 u_{2xy} + 3 u_{3y}) \right. 
\nonumber \\ 
&& \left. {} + 5 u_y (u_{3x} + 3 u_{x2y}) - 9 u_{xy} ( u_{2x} + u_{2y}) 
\right).
\nonumber
\end{eqnarray*}
The 2\,D homotopy operator formulas in (\ref{2Dhomotopy}) yield
\begin{eqnarray*}
\mc{H}_{\mf{u}(x,y)}^{(x)} E
&=& \int_0^1 \left( \mc{I}_{u(x,y)}^{(x)} E \right)
[\lambda \mf{u}]\,\frac{d \lambda}{\lambda} 
\nonumber \\ 
&=& \int_0^1 \left( 3 \alpha \lambda^3 u^4 
+ \beta \lambda^2 \left( 9 u^2 (u_{2x} + \tfrac{2}{3} u_{2y})
- 6  u (3 u_x^2 + u_y^2) \right) \right. 
\nonumber \\ 
&& {} + \tfrac{\beta^2}{\alpha} \lambda \left( 6 u_{2x}^2 + 5 u_{xy}^2
+ \tfrac{3}{2} u_{2y}^2 + \tfrac{3}{2} u (u_{2x2y} + u_{4y})
- u_x (12 u_{3x} + 7 u_{x2y}) \right. 
\nonumber \\ 
&& \left. \left. {} - u_y (3 u_{3y} + 8 u_{2xy})
+ \tfrac{5}{2} u_{2x} u_{2y} \right) \right)\,d \lambda 
\nonumber \\ 
&=& \tfrac{3}{4} \alpha u^4 
+ \beta \left( 3 u^2 (u_{2x} + \tfrac{2}{3} u_{2y})
- 2  u (3 u_x^2 + u_y^2) \right) 
\nonumber \\ 
&& {} + \tfrac{\beta^2}{\alpha} \left( 3 u_{2x}^2 + \tfrac{5}{2} u_{xy}^2
+ \tfrac{3}{4} u_{2y}^2 + \tfrac{3}{4} u (u_{2x2y} + u_{4y})
- u_x (6 u_{3x} + \tfrac{7}{2} u_{x2y}) \right. 
\nonumber \\ 
&& \left. {}  - u_y (\tfrac{3}{2} u_{3y} + 4 u_{2xy})
+ \tfrac{5}{4} u_{2x} u_{2y} \right),
\nonumber \\
\mc{H}_{\mf{u}(x,y)}^{(y)} E
&=& \int_0^1 \left( \mc{I}_{u(x,y)}^{(x)} E \right)
[\lambda \mf{u}]\,\frac{d \lambda}{\lambda} 
\nonumber \\ 
&=& \int_0^1 \left( 3 \lambda^2 \beta  u ( u u_{xy} - 4 u_x u_y)
- \tfrac{\beta^2}{2 \alpha} \lambda \left( 3 u ( u_{3xy} + u_{x3y})
+ u_x (3 u_{3y} + 13 u_{2xy}) \right. \right. 
\nonumber \\ 
&& \left. \left. {} + u_y (5 u_{3x} + 15 u_{x2y})
- \tfrac{9}{2} u_{2y} ( u_{2x} + u_{xy}) \right) \right)\,d \lambda 
\nonumber \\ 
&=&  \beta  u ( u u_{xy} - 4 u_x u_y)
- \tfrac{\beta^2}{4 \alpha} \left( 3 u ( u_{3xy} + u_{x3y})
+ u_x (13 u_{2xy} + 3 u_{3y}) \right. 
\nonumber \\ 
&& \left. {} + 5 u_y (u_{3x} + 3 u_{x2y}) 
- 9 u_{xy} ( u_{2x} + u_{2y}) \right).
\nonumber
\end{eqnarray*}
The flux $\mathbf{J} = \left( \mc{H}_{\mathbf{u}(x,y)}^{(x)} E,
\mc{H}_{\mathbf{u}(x,y)}^{(y)} E \right)$ has a curl term, 
$\mf{K} = (\D_y \theta, -\D_x \theta),$ with
\[
\theta 
= 2 \beta u^2 u_y + \tfrac{\beta^2}{4 \alpha} 
\Big(3 u (u_{2xy} + u_{3y}) 
+ 5 (2 u_x u_{xy} + 3 u_y u_{2y} + u_{2x} u_y) \Big).
\]
%
After removing $\mf{K}$ with the technique in Section \ref{secremovecurl},
the density-flux pair reads
\begin{eqnarray}
\label{ZKrank6conlaw}
\rho^{(3)} &=& u^3 - \tfrac{3\beta}{\alpha} (u_x^2 + u_y^2),
\nonumber \\ 
\mf{J}^{(3)} &=& \big( \tfrac{3 \alpha}{4} u^4 + 3 \beta u^2 u_{2x}
- 6 \beta u (u_x^2 + u_y^2)
+ \tfrac{3 \beta^2}{\alpha} (u_{2x}^2 - u_{2y}^2)
- \tfrac{6 \beta^2}{\alpha} (u_x (u_{3x} + u_{x 2y}) 
\nonumber \\ 
&& {} + u_y (u_{2x y} + u_{3y})),\,\, 3 \beta u^2 u_{xy}
+ \tfrac{6 \beta^2}{\alpha} u_{xy} (u_{2x} + u_{2y}) \big).
\end{eqnarray}
There is one additional density-flux pair \cite{Poole09} with explicit 
dependence on $x$ and $t,$ 
\begin{eqnarray}
\label{ZKrank1conlaw}
\rho^{(4)} &=& t u^2 - \tfrac{2}{\alpha} x u,
\nonumber \\ 
\mf{J}^{(4)} &=& \big( t (\tfrac{2 \alpha}{3} u^3 - \beta (u_x^2 - u_y^2)
+ 2 \beta u (u_{2x} + u_{2y}) ) - \tfrac{2}{\alpha} x ( \tfrac{\alpha}{2} u^2
+ \beta u_{2x} ) + \tfrac{2 \beta}{\alpha} u_x,  
\nonumber \\ 
&& {} \,\, - 2 \beta ( t u_x u_y + \tfrac{1}{\alpha} x u_{xy} ) \big),
\end{eqnarray}
which expresses the conservation of center of mass. 
No other density-flux pairs than the four reported above could be found 
\cite{Poole09} with {\tt ConservationLawsMD.m}.  
Zakharov and Kuznetsov \cite{ZakharovKuznetsov74} found three ``integrals of 
motion" which are identical to (\ref{ZKrank2conlaw}), 
(\ref{ZKrank4conlaw}), and (\ref{ZKrank1conlaw}). 
Without mention of fluxes, Infeld \cite{Infeld85} reports the densities in 
(\ref{ZKrank4conlaw}) and (\ref{ZKrank6conlaw}) and states that another
constant of motion exists, but only for the ZK equation in 3\,D.
Shivamoggi {\it et al}.\ \cite{Shivamoggi93} claimed that there are only 
four conservation laws for the ZK equation, but gave different results.  
Shivamoggi computed conservation laws for the potential ZK equation, 
$v_t + \tfrac{1}{2} v_x^2 + v_{3x} + v_{x2y} = 0,$
derived from (\ref{ZKequation}) with $\alpha = \beta = 1$ by setting 
$u = v_x$ followed by an integration with respect to $x.$
Doing so, he produced four nonlocal conservation laws for (\ref{ZKequation}).
The densities for the potential ZK equation reported in \cite{Shivamoggi93} 
are correct but the fluxes for three of the four conservation laws are
incorrect due to typographical errors \cite{Shivamoggi09}.

In summary, the existence of only four conservation laws confirms that 
(\ref{ZKequation}) is not completely integrable.
\vspace{-3mm}
\noindent
\section*{Acknowledgements}
This material is based in part upon research supported by the National Science 
Foundation (NSF) under Grant No.\ CCF-0830783.
Any opinions, findings, and conclusions or recommendations expressed in this
material are those of the authors and do not necessarily reflect the 
views of NSF. 

WH is grateful for the hospitality and support of the Department of 
Mathematics and Statistics of the University of Canterbury 
(Christchurch, New Zealand) during his sabbatical visit in Fall 2007.
M. Hickman (University of Canterbury) and B. Deconinck (University of
Washington) are thanked for sharing their insights about the scope and 
limitations of homotopy operators.

\end{document}